\documentclass[11pt]{amsart}

\usepackage[letterpaper, margin=1in]{geometry}  
\usepackage[parfill]{parskip}
\usepackage{amsmath}
\usepackage{amssymb}
\usepackage{amsfonts}
\usepackage{amsthm}
\usepackage[titletoc]{appendix}
\usepackage{hyperref}
\usepackage{tikz}
\usepackage[all]{xy}
\usepackage{stmaryrd}
\usepackage{indentfirst}
\usepackage{graphicx}
\usepackage{tabularx}
\usepackage{multirow}
\usepackage{footnote}
\usepackage{enumitem}
\usepackage{tikz-cd}
\usepackage{blkarray}
\makesavenoteenv{tabular}

\newcolumntype{Y}{>{\centering\arraybackslash}X}
\newcolumntype{L}{>{\centering\arraybackslash}m{3cm}}
\definecolor{light-gray}{gray}{0.9}

\theoremstyle{definition}
\newtheorem{thm}{Theorem}[section]
\newtheorem{prop}[thm]{Proposition}
\newtheorem{lemma}[thm]{Lemma}
\newtheorem{cor}[thm]{Corollary}
\newtheorem{defn}[thm]{Definition}
\newtheorem{defnprop}[thm]{Definition/Proposition}
\newtheorem{rmk}[thm]{Remark}
\newtheorem{ex}[thm]{Example}

\newtheorem{note}[thm]{Note}

\newcommand{\bb}[1]{\mathbb{#1}}
\newcommand{\mr}[1]{\mathrm{#1}}
\newcommand{\mc}[1]{\mathcal{#1}}

\newcommand{\mf}[1]{\mathfrak{#1}}
\newcommand{\del}{\partial}
\newcommand{\ol}{\overline}
\newcommand{\ul}{\underline}
\newcommand{\wt}{\widetilde}

\newcommand{\eps}{\varepsilon}

\newcommand{\llb}{\llbracket}
\newcommand{\rrb}{\rrbracket}

\newcommand*\circled[1]{\tikz[baseline=(char.base)]{ \node[shape=circle,draw,inner sep=2pt] (char) {#1};}}

\DeclareMathOperator{\delbar}{\bar{\del}}

\DeclareMathOperator{\perf}{Perf}

\DeclareMathOperator{\sym}{Sym}

\DeclareMathOperator{\tr}{Tr}

\DeclareMathOperator{\SL}{SL}
\DeclareMathOperator{\SU}{SU}

\DeclareMathOperator{\SHO}{SHO}

\DeclareMathOperator{\C}{\bb{C}}

\DeclareMathOperator{\Z}{\bb{Z}}
\DeclareMathOperator{\PV}{PV}

\DeclareMathOperator{\Vect}{Vect}

\usepackage[backend=bibtex,
            isbn=false,
            doi=false,
            giveninits=true,
            style=alphabetic,
            useprefix=true,
            maxcitenames=5,
            maxbibnames=5,
            sorting=nyt,
            labelnumber, defernumbers = true,
            citestyle=alphabetic]{biblatex}

\bibliography{references.bib}

\begin{document}

\title{Minimal models for minimal BCOV theories}

\author{Surya Raghavendran and Philsang Yoo}
\date{}

\maketitle

\begin{abstract}
Minimal BCOV theory is a classical field theory which describes a subclass of deformations of the category of perfect complexes on a Calabi--Yau variety. We compute minimal models for $L_\infty$-algebras describing minimal BCOV theory and its variants on flat space and  find that they give certain $L_\infty$-extensions of the infinite-dimensional simple Lie superalgebra $\SHO(d|d)$. We apply this computation to compare an $\mf{sl}_2$ action on an odd two-dimensional central of $\SHO(3|3)$ first discovered by Kac to an action of $\mf{sl}_2$ on a variant of minimal BCOV theory previously found by the authors.
\end{abstract}

\setcounter{tocdepth}{1}

\bigskip 

\tableofcontents

\section{Introduction}

In the seminal work \cite{BCOV}, Bershadsky, Cecotti, Ooguri, and Vafa studied a string field theory for the topological B-model on a Calabi--Yau 3-fold $X$. The equations of motion of the theory recover the equations of Kodaira--Spencer, describing deformations of the complex structure on $X$. The mathematical study of this theory was initiated by Costello and Li \cite{costello2012quantum} who constructed a larger theory which they called \textit{BCOV theory}, whose equations of motion describe deformations of the Calabi--Yau category $\perf X$ in arbitrary dimensions.

A salient feature of the work of \cite{costello2012quantum} is its use of a version of the classical Batalin--Vilkovisky (BV) formalism. In modern parlance, the classical BV formalism describes the structures present on formal neighborhoods of points in moduli spaces of solutions to equations of motion using tools from derived deformation theory. Indeed, derived deformation theory allows one to describe formal neighborhoods of points in moduli spaces in terms of the structure of an $L_\infty$-algebra on the shifted tangent complex \cite{dagx}, \cite{pridham}, \cite{hinich}. If the moduli space in question arises as solutions to some local PDE on a manifold, the brackets encoding the $L_\infty$-structure may further be taken to be poly-differential operators. Finally, the PDEs of classical field theory typically arise from some variational problem, which is codified in a shifted-symplectic structure on the corresponding formal moduli problem.

In \cite{costello2012quantum} it was noted that when articulated in the BV formalism, the space of fields does not quite have a shifted symplectic structure. Rather it has a shifted Poisson structure with a very large degeneracy. This reflects the nonlocal nature of the kinetic term in the action first introduced by \cite{BCOV}, and as explained by the authors of \cite{costello2012quantum}, arises from a variant of the construction of Barannikov--Kontsevich, which realizes the formal neighborhood of $\perf X$ in the moduli of Calabi--Yau categories as a nonlinear Lagrangian in a symplectic vector space \cite{barannikov2001quantum}, \cite{giventalcoates}.

The Poisson nature of the moduli problem motivated the definition of  \textit{minimal BCOV theory}, which we recall in Section \ref{minbcov} below. Morally, minimal BCOV theory can be thought of as akin to a symplectic leaf in the ambient formal moduli problem. Though the moduli-theoretic interpretation of this symplectic leaf is obscured -- it is no longer clear what sublcass of deformations of $\perf X$ the theory describes -- minimal BCOV theory has since appeared in a range of places, most notably in descriptions of twists of type I and type II supergravity in ten dimensions \cite{CostelloLi}, \cite{spinortwist}, \cite{Raghavendran_2023}.

Given an $L_\infty$-algebra, one may ask for a so-called \textit{minimal model}, which is a quasi-isomorphic $L_\infty$-algebra in which no differential appears. In this paper, we compute minimal models for minimal BCOV theory on flat space. We find that on $\C^d$ the minimal model is equivalent to a one-dimensional central extension of the Lie superalgebra $\SHO(d|d)$ that appears in Kac's classification \cite{Kac}. We further consider variants of minimal BCOV theory that have also featured in applications to twisted supergravity. The space of fields in BCOV theory includes divergence free holomorphic polyvector fields on a Calabi--Yau variety. A common trick in the literature is to assume that certain divergence free holomorphic polyvector fields are in fact total divergences -- we codify this trick in Definition \ref{def:bcovpot} and show that the minimal model of the underlying $L_\infty$-algebra on $\C^d$ is a further $L_\infty$-extension of this one-dimensional odd central extension of $\SHO(d|d)$. 

In the final section, we analyze a particular three-dimensional example, where the aforementioned $L_\infty$-extension becomes an ordinary extension as a Lie algebra, which further descends to $\SHO(3|3)$. The resulting odd two-dimensional central extension of $\SHO(3|3)$ was studied by Kac, Kac-Cheng, and Kac-Cantarini in \cite{Kac}, \cite{ChengKac}, \cite{CantariniKac}. In the first of these references, it was noted that $\SHO(3|3)$ and this odd two-dimensional central extension have an action of $\mf{sl}_2$ by outer derivations. We show how this action comes from a natural action on a variant of minimal BCOV theory in three complex dimensions, which we previously related to the action of S-duality on a twist of type IIB supergravity \cite{SuryaYoo}.

\subsection*{Conventions}
Throughout this article, we work in the BV formalism as articulated in the books of Costello and Gwilliam \cite{CG1}, \cite{CG2}. We will frequently work with $\Z/2$ graded BV theories whose underlying $L_\infty$-algebras have $\Z$-gradings. In these cases, we will often use notation that is reflective of the underlying $\Z$ grading, though it should ultimately only be remembered mod 2.

\subsection*{Acknowledgements}
We wish to thank Victor Kac and Nicoletta Cantarini for correspondence related to the content of Section \ref{sl2}. SR would like to thank Ingmar Saberi and Brian Williams for numerous discussions on topics in the same realm of ideas as this article. Philsang Yoo was supported by the National Research Foundation of Korea (NRF) grant funded by the Korea government(MSIT) (No. 2022R1F1A107114212).

\section{BCOV theory and its variants}

We begin by describing the set-up following Costello--Li \cite{costello2012quantum}. The readers are advised to consult the original article for more detail.

Let $(X, \Omega_X)$ be a $d$-dimensional Calabi--Yau manifold with a holomorphic volume form $\Omega_X$. We first describe the free-limit of the theory. Let $\PV^i(X)$ denote the sheaf of holomorphic sections of $\wedge^i T_X$ where $T_X$ is the holomorphic tangent bundle of $X$. This sheaf has a locally free resolution given by the complex from the space $\PV^{i,\bullet}(X) = \bigoplus_{j= 0}^d \PV^{i,j}(X)= \bigoplus_{j= 0}^d \Omega^{0,j}(X, \wedge^i  T_X )$ together with the $\delbar$ operator. Letting, $\PV(X) =  \bigoplus_{i,j} \PV^{i,j}(X)[-i-j]$, the wedge product and $\delbar$ operator equip $\PV (X)$ with the structure of a commutative DG algebra. 

Contracting with $\Omega_X$ yields an isomorphism of complexes \[(-) \vee \Omega_X \colon \PV^{i,\bullet}(X) \cong \Omega^{d-i,\bullet}(X),\] which we use to transport the holomorphic de Rham differential $\del$ on $\Omega ^{\bullet,\bullet}(X)$ to an operator on $\PV^{\bullet,\bullet}(X)$. The result $\del \colon \PV^{i,j}(X) \to \PV^{i-1,j}(X) $ is the divergence operator with respect to a holomorphic volume form $\Omega_X$. With this in hand, for each $i$, we can consider the subsheaf $\PV^i_\del (X) = \ker \del \subset \PV^i (X)$ of divergence free holomorphic $i$-polyvector fields on $X$. 

The sheaf of complexes underlying the free limit of BCOV theory is given by \[ \mc E_{\mr{BCOV}}(X) = (\PV(X)\llb t\rrb [2], Q= \delbar+t\del)\] where $t$ is a parameter of cohomological degree 2.  This complex provides a locally free resolution of the sheaf $\oplus_{i=0}^d \PV^i_\del (X)$ of divergence-free holomorphic polyvector fields on $X$.  The degree-shift arranges for the summand $\PV^{1,1}(X)$ to sit in degree 0; such elements describe deformations of the complex structure on $X$ and their presence in degree 0 articulates a sense in which BCOV is a theory of gravitational nature.

To specify a free theory in the BV formalism, one also specifies a nondegenerate local pairing of degree $-1$ on the space of fields, which equips the space of fields with the structure of $(-1)$-shifted symplectic space and hence, in particular, a Poisson bracket of cohomological degree 1 on its function space. However, BCOV theory fails to be a BV theory in the usual sense; this failure is already visible in the original formulation of the theory which involves a non-local kinetic term. Letting \[\tr \colon \PV_c(X)\to \bb C\qquad \text{denote} \qquad \mu\mapsto \int_X (\mu  \vee  \Omega_X) \wedge \Omega_X,\] where the subscript of $c$ denotes compactly supported sections, the kinetic term takes the form $S_{\mr{free}}(\mu) = \frac{1}{2}\tr (\mu \ol{\del} \del^{-1} \mu)$.

As observed in \cite{costello2012quantum}, such non-local kinetic terms can be dealt with using a variant of the BV formalism: BCOV theory is a so-called \textit{Poisson BV theory} in the sense of \cite{butson2016degenerate}. Accordingly, the theory instead has a degenerate copairing of odd degree specified by the constant coefficient bivector $(\del \otimes 1)\delta_{\mr{diag}}$ -- this equips the space of functions on the fields with the structure of a Poisson bracket of cohomological degree $2d-5$. As such, BCOV theory in 3-complex dimensions is a $\Z$-graded theory, but in all other dimensions, it is only $\Z /2$-graded.

The interactions of the theory are codified in the structure of a dg Lie algebra. Note that since the divergence operator is a second-order differential operator of cohomological degree $-1$, the expression \[ [\mu,\nu]_{\mr{SN}} := (-1)^{|\mu|-1} \left( \del (\mu  \nu) - (\del \mu)\nu - (-1)^{|\mu|} \mu  (\del \nu) \right)\] defines a Poisson bracket of degree $-1$, the familiar Schouten--Nijenhuis bracket on polyvector fields. Accordingly, $[-,-]_{\mr{SN}}$ satisfies a shifted version of graded Lie algebra axioms in that the following hold:
\begin{align*}
 [\alpha,\beta]_{\mr{SN}}& = -(-1)^{(|\alpha|-1) (|\beta|-1)} [\beta,\alpha]_{\mr{SN}} \\
 [\alpha,[\beta,\gamma]_{\mr{SN}} ]_{\mr{SN}} & =[[\alpha,\beta]_{\mr{SN}}, \gamma  ]_{\mr{SN}} +  (-1)^{ ( |\alpha|-1)(|\beta|-1) }   [\beta,[\alpha,\gamma ]_{\mr{SN}} ]_{\mr{SN}}.
\end{align*}

The Schouten--Nijenhuis bracket can be extended $t$-linearly to a graded Lie bracket on $\PV(X)\llb t \rrb$.  One can check $\del$ is a shifted derivation of $[-,-]_{\mr{SN}}$ in that \[\del [\mu,\nu]_{\mr{SN}} = [\del \mu,\nu]_{\mr{SN}} + (-1)^{|\mu|-1}[\mu, \del\nu]_{\mr{SN}}.\] As such, the Schouten--Nijenhuis bracket equips  $\mc E_{\mr {BCOV}}[-1]$ with the structure of a differential graded Lie algebra. 

Recall that a dg Lie structure or more generally an $L_\infty$-structure on a cochain complex $\mc E[-1]$ describes the Taylor components of an odd, square zero vector field on the formal moduli problem $\mc E$. Now if $\mc E$ describes a classical BV theory, this vector field in addition preserves the $(-1)$-shifted symplectic structure on $\mc E$. As symplectic vector fields on formal moduli problems are necessarily Hamiltonian, there is a local functional $S$ on $\mc E$, such that the vector field acts as $\{S, -\}$ where $\{-,-\}$ denotes the shifted Poisson bracket; then $S$ is the action functional of the theory. 

For Poisson BV theories like BCOV theory, the situation is more intricate. The vector field on $\mc E_{\mr{BCOV}}(X)$ corresponding to the dg Lie structure given by $(Q, [-,-]_{\mr{SN}})$ is not Hamiltonian. However, it is shown by Costello--Li \cite[Section 6]{costello2012quantum} that its pushforward under a transcendental automorphism is. Explicitly, there is a nontrivial $L_\infty$-map exhibiting an equivalence between the DG Lie structure $(Q=\delbar +t \del, [-,-]_{\mr{SN}} )$ on $\PV(X)\llb t\rrb [2]$ and an $L_\infty$-structure determined by a vector field of the form $Q+ \{I_{\mr{BCOV}},-\}$.  Here the ``interaction term'' $I_{\mr{BCOV}}$ is a local functional $I_{\mr{BCOV}}$ of cohomological degree $6-2d$ on $\mc E_{\mr{BCOV}}(X)$ that satisfies the classical master equation $QI + \frac{1}{2}\{I,I\}=0$. To define it, we first introduce the notation $\langle -\rangle_0 \colon \sym(\PV(X)\llb t \rrb ) \to \PV(X)$ for the descendent integral given by 
\[  \langle t ^{k_1}  \mu_1, t^{k_2} \mu_2,\cdots, t^{k_n}\mu_n\rangle_0 = \left(\int _{\ol{\mc M}_{0,n}} \psi_1^{k_1}\cdots \psi_n^{k_n} \right)  \mu_1 \cdots \mu_n = \binom{n-3}{k_1,k_2,\cdots,k_n} \mu_1 \cdots \mu_n. \]
Then one can succinctly write
\[I_{\mr{BCOV}}(\mu) = \tr \langle e^\mu\rangle_0.\]

We summarize the structures we have exposited in a definition.
\begin{defn}
A \textit{BCOV Theory} on a Calabi--Yau $d$-fold $X$ is the Poisson BV theory determined with:
\begin{itemize}
\item underlying sheaf of cochain complexes $(\PV(X)\llb t\rrb[2], Q =\ol\del + t \del)$
\item shifted Poisson structure given by the kernel $(\del\otimes 1)\delta_{\mr{diag}}$ with associated Poisson bracket $\{-,-\}$
\item an $L_\infty$-structure induced from the vector field $Q+\{I_{\mr{BCOV}},-\}$.
\end{itemize}
\end{defn}

\subsection{Minimal BCOV theory}\label{minbcov}

As BCOV Theory is equipped with a shifted Poisson structure from the kernel $(\del\otimes 1)\delta_{\mr{diag}}$, the Poisson structure pairs components of $\PV^{i,\bullet}(X)$ with $\PV^{d-i-1,\bullet}(X)$. In particular, the Poisson tensor has a very large degeneracy locus - it includes any field with a positive power of $t$. One would like to restrict to a subspace on which the Poisson structure is nondegenerate, but naively throwing away all fields in the degeneracy locus would not yield a cochain complex. Indeed, as the differential $\ol{\del} +t\del$ increases degree in $t$ the result would not be closed under the differential. This motivates the following definition
 
\begin{defn}\label{def:minbcov}
A \textit{minimal BCOV theory} $\mc E_{\mr{mBCOV}}(X) \subset \mc E_{\mr {BCOV}}(X)$ is defined by asking the space of fields to be given by the sheaf of subcomplexes \[\mc E_{\mr{mBCOV}}(X)=\left( \bigoplus_{i+ k\leq d-1}t^{k}\PV^{i,\bullet}(X), \ol{\del}+t\del\right)  \]
and setting $I_{\mr{mBCOV}}=I_{\mr{BCOV}}|_{\mc E _{\mr{mBCOV}}}$. 
\end{defn}

\begin{ex}[Minimal BCOV theory for Calabi--Yau 3-folds]
One has the space of fields
\[\xymatrix @R=0.5em {
 \ul{-2} &   \ul{-1} &   \ul{0}&   \ul{1}&   \ul{2}\\
 \PV^{0,\bullet}(X) \\
& \PV^{1,\bullet}(X) \ar[r]^-{t\del } & t\PV^{0,\bullet}(X)\\
 &&\PV^{2,\bullet}(X)  \ar[r]^-{t\del }  & t\PV^{1,\bullet}(X) \ar[r]^-{t\del }  & t^2 \PV^{0,\bullet}(X)
}\]
with the action functional given by
\begin{align*}
 I_{\mr{mBCOV}} &=   \tr \left\langle  \left(\alpha^0 \mu^1 \eta^2 + \frac{1}{6} \mu^1\mu^1\mu^1  + \frac{1}{2} \alpha^0 \alpha^0 \eta^2 (t \eta^1) + \frac{1}{2} \alpha^0\alpha^0 \mu^1 \eta^2 (t^2 \eta^0) \right) e^{t \mu^ 0}   \right\rangle_0  \\
&=  \sum_{k\geq 0 } \tr   \left( \alpha^0 \mu^1 \eta^2 (\mu^0)^k+ \frac{1}{6}\mu^1\mu^1\mu^1   (\mu^0)^k \right)	 + \tr \left\langle  \left( \frac{1}{2} \alpha^0 \alpha^0 \eta^2 (t \eta^1) + \frac{1}{2} \alpha^0\alpha^0 \mu^1 \eta^2 (t^2 \eta^0) \right) e^{t \mu^ 0}   \right\rangle_0
\end{align*}
where $\alpha^0 \in \PV^{0,\bullet}(X)$, $(\mu^1, t \mu^0)\in \PV^{1,\bullet}(X) \oplus  t \PV^{0,\bullet}(X)$, and $(\eta^2, t \eta^1, t^2 \eta^0) \in \PV^{2,\bullet}(X) \oplus t \PV^{1,\bullet}(X) \oplus t^2 \PV^{0,\bullet}(X)$.
\end{ex}

\subsection{Introducing potentials}

Now we introduce variants of minimal BCOV theory for $d\geq 3$. Let $2 \leq k \leq d-1$. The motivation behind our variant is to replace the summand $\left(\bigoplus_{i+ j= k}t^{i}\PV^{j,\bullet}(X) , \ol\del +t\del \right)$, which resolves divergence free holomorphic $k$-polyvectors, with a resolution of the complex of holomorphic $(k+1)$-polyvectors considered up to a total divergence.

 That is, consider the natural cochain map 
 \[
\del \colon \left(\xymatrix{t^{k-d+1}\PV^{d,\bullet}(X)\ar[r] & \cdots \ar[r] & \PV^{k+1,\bullet}(X)  }   \right)  \to \left(\xymatrix{\PV^{k,\bullet}(X) \ar[r] & \cdots \ar[r] & t^{k} \PV^{0,\bullet}(X) }  \right).
\]
Note that the source of this map has negative degrees of the parameter $t$; the differential is again $t\del$. In the theory we wish to describe, $\PV^{k+1,\bullet}(X)$ will live in the same cohomological degree as $\PV^{k,\bullet}(X)$, with parity differing from that of a usual element of polyvector fields. The interactions will be given by pulling back $I_{\mr{mBCOV}}$ along this map.

\begin{defnprop}\label{def:bcovpot}
Let $X$ be a Calabi--Yau $d$-fold, and fix $2\leq k \leq d-1$ such that when $d$ is odd, $k\neq \frac{d-1}{2}$. \textit{Minimal BCOV theory on $X$ with $k$-potentials} $\wt{\mc E}^k_{\mr{mBCOV}}(X)$ is the Poisson BV theory with:
\begin{itemize}
\item underlying space of fields given by the sheaf of cochain complexes
\[
\left( \bigoplus_{0\leq i \leq d-k-1}t^{-i}\PV^{k+i+1,\bullet}(X)[2-k+i] \oplus \bigoplus_{\substack{i+j\leq d-1\\ i+j\neq k}}t^{i}\PV^{j,\bullet}(X), \ol{\del}+t\del\right)
\]
\item shifted Poisson structure given as follows: the pairing between the summands $\oplus_{0\leq i \leq d-k-1}t^{-i}\PV^{k+i+1,\bullet}(X)$ and $\oplus_{0\leq i\leq d-k-1} t^i \PV^{d-k-i-1,\bullet}(X)$ is non-degenerate and is given by the trace map $\operatorname{Tr}$. On all other summands, the shifted Poisson structure is given by the bivector $(\del \otimes 1 )\delta_{\mr{diag}}$.

\item an action functional given by $\wt I_{\mr{mBCOV}} = \Phi^* I_{\mr{mBCOV}} $, where $\Phi : \wt{\mc E}^k_{\mr{mBCOV}}(X) \to \mc E_{\mr{mBCOV}}(X)$ is the map given by extending the map $\del$ by the identity on all other summands.
\end{itemize}
\end{defnprop}

\begin{proof}
Note that the map $\Phi$ is trivially a Poisson map. Hence $\wt I_{\mr{mBCOV}}$ satisfies the classical master equation as $I_{\mr{mBCOV}}$ does.
\end{proof}

\begin{rmk}
The case where $d$ is odd and $k =\frac{d-1}{2}$ was thrown out by hand because such fields pair with themselves in minimal BCOV theory. In this case, the sheaf of complexes $\wt{\mc{E}}^k_{\mr{mBCOV}}(X)$ can still be equipped with the structure of a pre-symplectic BV theory \cite{presymplecticbv}.  We expect that theorem \ref{thm:bcov_d}, which is a result about the underlying $L_\infty$ structure, will still hold in that case.
\end{rmk}

Such theories appear in the context of twists of supergravity, see \cite{SuryaYoo}, \cite{spinortwist}, \cite{Raghavendran_2023}, and in particular section 8.2 of \cite{CostelloLi} for motivation.

\section{Minimal models and extensions of $\SHO(d|d)$}

In the Batalin--Vilkovisky formalism, the parity shifted space of fields always has the structure of an $L_\infty$-algebra. Moreover, given an $L_\infty$-algebra, one can ask for a so-called minimal model, which is a quasi-isomorphic $L_\infty$-algebra where the differential is zero. Minimal models can be computed by way of homotopy transfer \cite{operadsBook} and physically describe tree-level scattering amplitudes in the theory \cite{perturbiner}, \cite{wolf}.

In this paper, we wish to compute minimal models of minimal BCOV theory and some of its variants. We claim that the results can be described in terms of extensions of the infinite-dimensional simple Lie superalgebra $\SHO(d|d)$. Following \cite{ChengKac},  $\SHO(d|d)$ may naturally be constructed as a Lie subalgebra of the Lie algebra $\Vect(\C^{d|d})$ of vector fields on the super-vector space $\C^{d|d}$. Choosing holomorphic coordinates $x_i, \xi_j$ on $\C^{d|d}$, we may write a vector field $\mu \in \Vect (\C^{d|d} )$ as 
\[
\sum_{i=1}^d \mu_{x_i}\frac{\del}{\del x_i} + \sum_{i=1}^d \mu_{\xi_i} \frac{\del}{\del \xi_i}
\]
where $\mu_{x_i}, \mu_{\xi_j}\in \mc O (\C^{d|d} )$. 

Note that $\C^{d|d}$ also comes equipped with the following two pieces of structure:
\begin{itemize}
\item It possesses a canonical odd-shifted symplectic form -- in terms of the coordinates we have chosen, it can be written as $\omega = \sum_{i=1}^d d x_i \wedge d \xi_i$. 

\item There is a canonical volume form on $\C^{d|d}$. Using this, we can define a super-divergence operator
\[D : \Vect (\C^{d|d}) \to \mc O (\C^{d|d})\]
which, in terms of the coordinates we have chosen, is given by 
\[D \mu = \sum_{i=1}^{d} \frac{\partial \mu_{x_i}} {\partial x_i} + \sum _{i=1}^{d} (-1)^{|\mu_{\xi_i} | }\frac{\partial \mu_{\xi_i}}{\partial \xi_i}\]
\end{itemize}

In terms of this data, the following Lie superalgebras appear in Kac's classification:

\begin{defn}
\hfill 
\begin{itemize}
\item Let $\operatorname{HO}(d|d) = \{ \mu \in \Vect ( \C^{d|d} ) \mid  \mc L_\mu \omega = 0 \}$ denote the Lie superalgebra of symplectic vector fields on $\C^{d|d}$.
\item Let $\SHO^\prime (d|d) = \{ \mu \in \operatorname {HO} (d|d) \mid D\mu = 0\} $ denote the sub Lie superalgebra consisting of those symplectic vector fields that are also super-divergence-free. 
\item Let $\SHO (d|d)$ denote the derived subalgebra of $\SHO^\prime (d|d)$. 
\end{itemize}
\end{defn}

\begin{note}\label{strat}
In proposition\ref{prop:sho} and theorems \ref{thm:bcov_d}, \ref{thm:bcov_k}, we compute minimal models for the $L_\infty$-algebras determined by definitions \ref{def:minbcov}, \ref{def:bcovpot} respectively. These results all follow the same general pattern, so we record the steps here to help orient the reader. Let $(\Pi \mc E, Q)$ be one of the $L_\infty$-algebras appearing in definitions. The minimal model is computed using the following steps

\begin{itemize}
	\item [\circled{1}] Compute $H^\bullet \left ( \Pi \mc E(\C^d); Q \right )$ as a graded vector space. Note that the complex $(\Pi \mc E(\C^d), Q)$ is naturally the totalization of a bicomplex where one of the differentials is the $\delbar$ operator, so we may compute cohomology by way of the associated spectral sequence. 
		
	\item [\circled{2}] Compute the $L_\infty$-structure on $H^\bullet \left (\Pi \mc E (\C^d); Q \right )$ by way of homotopy transfer. To do so, we make use of the results of \cite{Cirici2022} which combines spectral sequences with homotopy transfer. Indeed, their results give a sequence of $L_\infty$-structures on each page of the associated spectral sequence, which are all related by homotopy transfer. A minimal model for $(\Pi\mc E(\C^d), Q)$ is gotten by the $L_\infty$-structure on the page where the spectral sequence degenerates.
		
In each case, we will find that there is no nontrivial transfer onto the $E_1$-page, and to contemplate transfer onto the $E_2$-page, we will need to find explicit homotopy data and analyze the possibly nonzero diagrams. 

	\item [\circled{3}] Find an isomorphism between $H^\bullet \left ( \Pi \mc E(\C^d); Q \right )$ with the transferred $L_\infty$-structure and an extension of $\SHO(d|d)$.
\end{itemize}
\end{note}

\subsection{Minimal BCOV theory and $\SHO(d|d)$}

\begin{prop}\label{prop:sho}
The minimal model for minimal BCOV theory $\Pi\mc E_{\mr{mBCOV}} (\C^d)$ is a one-dimensional odd central extension of $\SHO (d | d)$.
\end{prop}

\begin{proof}

\circled{1} We first compute $H^\bullet \left ( \mc E_{\mr{mBCOV}} (\C^d); Q \right )$ as a graded vector space. This cohomology can be computed by way of a spectral sequence whose first page is the cohomology with respect to the Dolbeault differential. By the $\delbar$-Poincare lemma, the first page of this spectral sequence is given by the graded vector space $\oplus_{i+j\leq d-1} t^i\PV^{j} (\C^d)$. The differential on the first page is the divergence operator $\del$, and we see that the cohomology with respect to the totalized differential is given by $\oplus_{i\leq d-1} \PV^i_\del (\C^d)$. 

\circled{2} 
We now wish to find homotopy data to compute the transferred $L_\infty$-structure on $H^\bullet \left ( \Pi  \mc E_{\mr{mBCOV}} (\C^d); Q \right )$. We organize our computation according to the spectral sequence used in the previous step. We claim that there can be no nontrivial homotopy transfer onto the first page. Indeed, note that the $L_\infty$-structure on $\mc E_{\mr{mBCOV}}$ comes from the tensor product of an $L_\infty$-algebra and the natural cdga structure on $\Omega^{0,\bullet} (\C^d)$. Any transfer of $L_\infty$-structure onto the first page of this spectral sequence must therefore be induced by a transfer of a $C_\infty$-structure onto $H^\bullet (\Omega^{0,\bullet}(\C^d); \delbar)\cong \mc O (\C^d)$. However, no higher brackets can appear for form-type reasons. Therefore, we need only consider transfer onto the $\del$-cohomology.

We claim that the following is a homotopy datum:
\[
\begin{tikzcd}
\ar [loop left]{l}{H} \left (\Pi \bigoplus_{i+j\leq d-1} t^i\PV^{j} (\C^d), t\del \right )\ar [r, shift left, "p"] &H^\bullet \left ( \mc E_{\mr{mBCOV}} (\C^d); Q \right ) \ar[l, shift left, "\iota"] \: ,
\end{tikzcd}
\]
\begin{itemize}
\item  For $1\leq j \leq d-1$, let \[K^j: \PV^{j} \to \PV^{j+1}\] be such that $\del K^j + K^{j-1} \del =\operatorname{id}$. The explicit form of this operator will not be needed; such an operator exists by the proof of the holomorphic Poincare lemma. We will further require that $K^{d-1} : \PV^{d-1} \to \PV^d$ produces a top polyvector field with vanishing constant term; this is always possible. We define the homotopy $H$ to act as zero on summands of the form $t^0 \PV^j$ and as $t^{-1}K$ on all other summands. Note that for those descendants of the form $t^k \PV^0$ for $k\neq 0$, only the first term in this relation contributes.

\item The projection $p$ is defined as follows
\[ p(t^i\mu^j) = \begin{cases}
  0&  \text{if }i\neq 0\\
 \mu^j - H (t\del) \mu^j & \text{if }i=0, \ j\neq 0\\
  \mu^0  & \text{if }i=j=0	
 \end{cases}
 \]
To see that these assignments land in $\oplus_{i\leq d-1} \PV^i _\del(\C^d)$, it is enough note that
\[ 
\del (\mu^j - H (t\del) \mu^j) = \del \mu^j - \del \mu^j + K\del^2 \mu^j = 0.
\]
\item The operator $\iota$ is simply the inclusion $\oplus_{i\leq d-1} t^0\PV^i_\del (\C^d)\to \oplus_{i\leq d-1} t^0\PV^i (\C^d)$.
\end{itemize}

We claim that the above operators satisfy the desired relations. It is clear that $p\iota =\operatorname{id}$. For the relation $\operatorname{id} - \iota p = (t\del) H + H(t\del)$ there are three cases to check. 
\begin{itemize}
\item First consider elements of the form $t^i\mu^j \in t^i \PV^{j}(\C^d)$ for $i\neq 0$. The desired relation then obviously follows because of how the operator $K$ was chosen.

\item Next consider elements of the form $\mu^j \in t^0 \PV^{j} (\C^d)$ for $j\neq 0$. We have that 
\[
(\operatorname{id} - \iota p) (\mu^j) = \mu^{j} - \mu^{j} + H(t\del)\mu^j = K\del \mu^j
\]
Similarly, $H$ vanishes on elements in $t^0 \PV^{j} (\C^d)$ so we have that $(t\del H + H t\del)\mu^j = K\del \mu^j$. 

\item Finally consider elements of the form $\mu^0 \in t^0 \PV^{0} (\C^d)$. We have that
$(\operatorname{id} - \iota p) \mu^0 = 0$. Similarly, $(t\del H + H t\del)\mu^0$ because $H$ acts as zero on such elements. 

\end{itemize}
This completes the check that the above assignments comprise valid homotopy data.

Now we analyze which diagrams can possibly be nonzero. Recall that the $L_\infty$-brackets on $\mc E_{\mr{mBCOV}}$ come from the Taylor components of the vector field $Q + \{I_{\mr{mBCOV}},- \}$. The fact that the shifted Poisson tensor vanishes on descendants simplifies the kinds of brackets that can appear. For instance, no bracket can output a descendant, which readily implies that any nonzero diagram cannot involve any homotopies. 

Therefore, the only nonzero diagrams consist of a single vertex -- we claim that it in fact must be cubic. Indeed, a vertex of quartic or higher order involves a bracket which necessarily takes a descendant as input, but the image of $\iota$ does not involve any descendants. 

Thus we see that the cohomology is simply an ordinary Lie superalgebra. The Lie bracket given by nonzero diagrams is determined by the cubic term in the action and is given by  $[\alpha, \beta] = \del (\alpha \wedge \beta)$ up to a sign. Because $\alpha, \beta\in \ker \del$, this bracket coincides with the Schouten--Nijenhuis bracket.

\circled{3} To finish the proof, we want to identify an explicit isomorphism between $H^\bullet \left ( \mc E_{\mr{mBCOV}} (\C^d); Q \right )$ and $\SHO (d|d) \oplus \Pi \C$ as Lie superalgebras. This step is essentially a reproduction of an alternative characterization of $\SHO (d|d)$ following the discussion in Section 1.3 of \cite{ChengKac}

Note that the odd-symplectic structure on $\C^{d|d}$ naturally equips $\mc O(\C^{d|d})$ with the structure of an odd-shifted Poisson algebra, hence equipping $\Pi \mc O (\C^{d|d} )$ with a Lie superalgebra structure. In fact, we may identify $\mc O (\C^{d|d}) \cong \PV (\C^d )$, and doing so, this Lie superalgebra structure is exactly the Schouten--Nijenhuis bracket. 

Accordingly, there is a short exact sequence of Lie superalgebras
\[
0 \to \C \to \Pi \mc O( \C^{d|d} ) \to \operatorname{HO}(d|d ) \to 0
\]
which is a super-version of the familiar short exact sequence expressing Hamiltonian vector fields on a manifold with trivial $H^1$ as a central extension of symplectic vector fields. Explicitly, the map $\Pi \mc O(\C^{d|d} ) \to \operatorname{HO} (\C^{d|d} )$ is given by $f  \mapsto \Theta (df, - )$ where $\Theta$ and $d$ denote the canonical odd-Poisson tensor and de Rham differential on $\C^{d|d}$ respectively. In terms of the coordinates we have chosen, the map is given by \[f (x_i, \xi_j) \mapsto \sum_{i=1}^d \frac{\del f}{\del x_i} \frac{\del}{\del \xi_i} + (-1)^{|f|} \frac{\del f}{\del \xi_i} \frac{\del}{\del x_i}.\] 
Choosing a splitting, we may therefore identify $\operatorname{HO}(d|d) \oplus \C \cong \Pi  \PV (\C^d)$. 

Now note that there is a commutative diagram with exact rows and columns
\[
\begin{tikzcd}
 &  & \Pi \mc O(\C^{d|d} )\ar[equal]{r}  & \Pi \mc O (\C^{d|d} )  & \\
0\ar[r]  & \C\ar[r]  & \Pi \mc O(\C^{d|d} )\ar[r]\ar[u, "\del"]  & \operatorname{HO}(d|d)\ar[r]\ar[u, "D"] & 0 \\
0\ar[r]  & \C\ar[r]\ar[equal]{u}  & \ker \del\ar[r]\ar[u]  & \operatorname{SHO}^\prime (d|d)\ar[r]\ar[u] & 0 \\
& & 0 \ar[u] & 0\ar[u]& \\
\end{tikzcd}
\]
where the commutativity of the top square is established via an explicit, local computation. In particular, under the identification $\mc O (\C^{d|d} ) \cong \PV(\C^d)$, divergence-free polyvector fields on $\C^d$ are exactly those symplectic vector fields on $\C^{d|d}$ with vanishing super-divergence, up to constants. 

Consequentially, a choice of splitting lets us identify $\operatorname{SHO}^\prime (d|d) \oplus \C \cong \ker \del$. The derived subalgebra $\operatorname{SHO}(d|d)$ consists of those elements of $\operatorname{SHO}^\prime (d|d)$ not containing the monomial $\xi_1\cdots \xi_d$. The cocycle exhibiting $\ker \del$ as a one-dimensional central extension of $\operatorname{SHO}^\prime (d|d)$, pulled back to the derived subalgebra, now exhibits $H^\bullet \left ( \mc E_{\mr{mBCOV}} (\C^d); Q \right )$ as a one-dimensional central extension of $\operatorname{SHO}(d|d)$.

\end{proof}

\subsection{Extensions from potentials}

In the case $k=d-1$, minimal BCOV theory with $k$-potentials on $\C^d$ has an ordinary Lie superalgebras as a minimal model.

\begin{thm}\label{thm:bcov_d}
Minimal BCOV theory with $(d-1)$-potentials $\Pi \wt {\mc{E}}^{d-1} _{\mr{mBCOV} }(\C^d)$ is $L_\infty$-equivalent to an extension of $\SHO(d|d)$ by:
\begin{itemize}
\item $\Pi \C^2$ when $d$ is odd.
\item $\C^{1|1}$ when $d$ is even.
\end{itemize}
\end{thm}

\begin{proof}
We again follow the strategy outlined in Note \ref{strat}.

\circled{1}
Once again using a spectral sequence whose first page is given by the cohomology with respect to the $\delbar$ operator, we see that 
\[
H^\bullet \left (\Pi \wt {\mc{E}}^{d-1} _{\mr{mBCOV} }(\C^d); Q \right )\cong \PV^d (\C^d)[2-d] \oplus \bigoplus_{i\leq d -2} \PV^i_\del (\C^d)
\]
as graded vector spaces.

\circled{2}
Next, we compute the induced $L_\infty$-structure on $H^\bullet (\Pi \wt {\mc{E}} _{\mr{mBCOV} }(\C^d); Q)$ via homotopy transfer.

We again organize the computation according to the spectral sequence of the previous step. It's clear that there are no higher brackets that can occur in homotopy transfer onto $\delbar$ cohomology for form-type reasons.

We next claim that the following is a homotopy datum:
\[
\begin{tikzcd}
\ar [loop left]{l}{H} \Pi \left (\PV^d (\C^d)[3-d] \oplus \bigoplus_{i+j\leq d-2} t^i\PV^{j} (\C^d), t\del \right )\ar [r, shift left, "p"] & \ar[l, shift left, "\iota"] \ H^\bullet \left (\Pi \wt {\mc{E}}^{d-1} _{\mr{mBCOV} }(\C^d); Q \right ) ,
\end{tikzcd}
\]
\begin{itemize}
\item The homotopy $H$ is once again defined by use of the auxiliary operator $K^i$ used in the proof of Proposition \ref{prop:sho}.  We take $H$ to act as zero if $i = 0$ and as $t^{-1}K$ on all other summands. 
\item The operator $\iota$ acts as the identity on $\PV^{d}(\C^d)$, and on every other summand, is just the natural inclusion $t^0 \PV^i_\del (\C^d) \to t^0 \PV^i (\C^d)$.
\item The operator $p$ is defined as follows
\[ p(t^i\mu^j) = \begin{cases}
  0&  \text{if }i\neq 0\\
 \mu^j - H \del \mu^j & \text{if }i=0, \ j\neq 0, d\\
  \mu^j  & \text{if }i=0, \ j = 0, d	
 \end{cases}
 \]
\end{itemize}
The proof that this constitutes valid homotopy data is identical to the proof of Proposition \ref{prop:sho}, and identical arguments show that the only nonzero diagram is a cubic diagram with a single vertex. There are three different kinds of brackets coming from terms in the action which involve one of $\mu^0$ or $\mu^d$, both of them, or none. 

\begin{enumerate}
\item 
For $a,b$ such that $a, b, a+b \neq 0, d$, there is a bracket
\begin{align*}
\PV^{a}_\del \times \PV^{b}_\del &\to \PV^{a+b-1}_\del \\
(\mu^a,\mu^b) &\mapsto \del(\mu^a \wedge\mu^b)
\end{align*}
coming from terms in the action of the form  $\operatorname{Tr}(\mu^a\mu^b \mu^c)$ where none of $a, b, c$ is $0$ or $d$ yield a bracket given by $\del (\mu^a\wedge \mu^b)$.

\item For $a,b$ such that $a+b= d$, there is a bracket
\begin{align*}
\PV^{a}_\del \times \PV^{b}_\del &\to \PV^{d} \\
(\mu^a,\mu^b) &\mapsto \mu^a \wedge\mu^b
\end{align*}
coming from terms in the action of the form $\operatorname{Tr}(\mu^0\mu^a \mu^b)$ where $a+b = d$.

\item Lastly, there is a bracket of the form 
\begin{align*}
\PV^{1}_\del \times \PV^{d} &\to \PV^{d} \\
(\mu^1,\mu^d) &\mapsto \mu^1 \wedge\del\mu^d
\end{align*}
 coming from the term in the action of the form $\operatorname{Tr}(\mu^0\mu^1 \del \mu^d)$ which comes from pulling back a term in $I_{\rm{mBCOV}}$ involving $\mu^{d-1}$. 
 \end{enumerate}

\circled{3}
Finally, we show that the cohomology $H^\bullet (\wt {\mc{E}} _{\mr{mBCOV} }(\C^d); Q)$ with the induced Lie superalgebra structure is a two-dimensional central extension of $\SHO(d|d)$. 
For the first part, consider the following short exact sequence of Lie algebras
\[
0\to \ker\Phi_*\to  H^\bullet (\Pi \wt {\mc E}_{\mr {mBCOV}}(\C^d); Q) \stackrel{\Phi_*}{\to} H^\bullet (\Pi {\mc E}_{\mr {mBCOV}}(\C^d); Q)\to 0.
\] 
The map $\Phi_*$ is induced from $\Phi$, which is in turn given by the $\del$ operator on the summand given by $\PV^{d,\bullet}$ and the identity elsewhere. It is clear that $\ker \Phi_*$ consists of those elements of $\PV^{d}$ which are constant coefficient; if $d$ is even this sits in even degree and if $d$ is odd it sits in odd degree. Therefore, using the splitting provided by the operator $K^{d-1}|_{\ker \del}$, we have an isomorphism
\[
H^\bullet (\Pi \wt {\mc E}_{\mr {mBCOV}}(\C^d); Q) \cong H^\bullet (\Pi {\mc E}_{\mr {mBCOV}}(\C^d); Q) \oplus (\Pi) \C.
\]

The corresponding cocycle is given by
\begin{align*}
H^\bullet (\Pi {\mc E}_{\mr {mBCOV}}(\C^d); Q)\otimes H^\bullet (\Pi {\mc E}_{\mr {mBCOV}}(\C^d); Q) &\to (\Pi )\C \\
[\alpha]\otimes [\beta] & \mapsto (\alpha\wedge\beta)(0) \vee \Omega
\end{align*}
which just gives the constant term in the coefficient of $\alpha\wedge\beta$. Note that this cocycle is nonzero for $\alpha\in t^0\PV^{i, 0}$ and $\beta\in t^0 \PV^{d-i, 0}$, $i = 1, \cdots, d-1$.

By Proposition \ref{prop:sho} we can identify $H^\bullet (\Pi {\mc E}_{\mr {mBCOV}}(\C^d); Q)\cong \SHO(d|d) \oplus \Pi \C$, so we find that $H^\bullet (\wt \Pi {\mc E}_{\mr {mBCOV}}(\C^d); Q)$ is a one-dimensional central extension of $\SHO(d|d) \oplus \Pi \C$. However, note that the above cocycle is in fact pulled back from a cocycle on $\SHO(d|d)$ -- it does not depend on the odd one-dimensional center, which under the identification with $H^\bullet (\Pi {\mc E}_{\mr {mBCOV}}(\C^d); Q)$, comes from the constants in $\PV^0$. Therefore, we see that 
\begin{itemize}
\item $H^\bullet (\Pi \wt {\mc E}_{\mr {mBCOV}}(\C^d); Q) \cong \SHO(d|d) \oplus \Pi \C^2$ when $d$ is odd
\item $H^\bullet (\Pi \wt {\mc E}_{\mr {mBCOV}}(\C^d); Q) \cong \SHO(d|d) \oplus \C^{1|1}$ when $d$ is even
\end{itemize}
as claimed.
\end{proof}

\begin{thm}\label{thm:bcov_k}
Suppose in addition that $k\neq d-1$. The minimal model for minimal BCOV theory with $k$-potentials $\Pi \wt {\mc{E}}^k _{\mr{mBCOV} }(\C^d)$ is an $L_\infty$-extension of $\Pi\C \oplus \SHO(d|d)$ by:
\begin{itemize}
\item $\Pi \C$ when $d$ is odd
\item $\C$ when $d$ is even
\end{itemize}
The $L_\infty$-structure has quadratic and $(d-k+1)$-ary brackets.
\end{thm}

\begin{proof}
We again follow the strategy outlined in Note \ref{strat}.

\circled{1}
Once again using a spectral sequence whose first page is given by the cohomology with respect to the $\delbar$ operator, we see that 
\[
H^\bullet \left (\Pi \wt {\mc{E}}^k _{\mr{mBCOV} }(\C^d); Q \right )\cong \C[2-d] \oplus \PV^{k+1} (\C^d)/\del \PV^{k+2} (\C^d)[1-k] \oplus  \bigoplus_{\substack{i\leq d -1\\ i \neq k}} \PV^i_\del (\C^d)
\]
as graded vector spaces. Here, the degree shifts are written in a $\Z$-graded way but are only taken mod $2$. 

\circled{2}
Next, we compute the induced $L_\infty$-structure on $H^\bullet (\Pi \wt {\mc{E}} _{\mr{mBCOV} }(\C^d); Q)$ via homotopy transfer. We again organize the computation according to the spectral sequence of the previous step. It's clear that there are no higher brackets that can occur in homotopy transfer onto $\delbar$ cohomology for form-type reasons. 

We next claim that the following is a homotopy datum:
\[
\begin{tikzcd}
\ar [loop left]{l}{H}\left( \bigoplus_{0\leq i \leq d-k-1}t^{-i}\PV^{k+i+1}(\C^d) \oplus \bigoplus_{\substack{i+j\leq d-1\\ i+j\neq k}}t^{i}\PV^{j}(\C^d), t\del\right)\ar [r, shift left, "p"] & \ar[l, shift left, "\iota"] \  H^\bullet \left (\Pi \wt {\mc{E}}^k _{\mr{mBCOV} }(\C^d); Q \right ),
\end{tikzcd}
\]

\begin{itemize}
\item The homotopy $H$ is once again defined by use of the auxiliary operator $K^i$ used in the proof of Proposition \ref{prop:sho}.  We take $H$ to act as zero on the summands $t^0 \PV^{j}$ and $t^{k-d+1}\PV^d$. The operator $H$ acts as $t^{-1}K$ on all other summands. 
\item The operator $\iota$ is defined as follows. On the summand $\bigoplus_{\substack{i\leq d -1\\ i \neq k}} \PV^i_\del (\C^d)$, $\iota$ is just the natural inclusion $\bigoplus_{\substack{i\leq d -1\\ i \neq k}} \PV^i_\del (\C^d) \to \bigoplus_{\substack{j\leq d-1\\ j\neq k}}t^{0}\PV^{j}$. On the summand $\PV^{k+1}/\del \PV^{k+2}$, the operator $\iota$ is defined by $[\mu^{k+1}]\mapsto K\del \mu^{k+1}$. Note that this is independent of the choice of representative of $[\mu^{k+1}]$. Lastly the map $\iota$ embeds $\C$ as constant coefficient elements of $t^{k-d+1}\PV^d$. 
\item The operator $p$ is defined as follows
\[ p(t^i\mu^j) = \begin{cases}
  0&  \text{if }i\neq 0, k-d+1\\
 \mu^j - H (t\del) \mu^j & \text{if }i=0, \ j\neq 0, {k+1}\\
  \mu^j  & \text{if }i=0, \ j = 0 \\
  [\mu^{j}] & \text{if }i=0, \ j = k+1 \\
\mu^d (0)\vee \Omega & \text{if }i ={k-d+1}
 \end{cases}
 \]
\end{itemize}
The proof that this constitutes valid homotopy data is similar to the proof of Proposition \ref{prop:sho}. The only two cases that differ from the argument there are 
\begin{itemize}
\item If $\mu^{k+1}\in \PV^{k+1}$, then we have that 
\[
(\operatorname{id} - \iota p) \mu^{k+1} = \mu^{k+1} - K\del \mu^{k+1}
\]
Now, the right-hand side of the homotopy relation only has one term because no differential leaves the summand $\PV^{k+1}$. 
\[
((t\del)H + H(t\del))\mu^{k+1} = \del K\mu^{k+1}.
\]
The homotopy relation now holds because of how the operator $K$ was chosen.
\item If $t^{k-d+1}\mu^d \in t^{k-d+1}\PV^d$, then we have that
\[
(\operatorname{id} - \iota p) t^{k-d+1}\mu^d = \mu^d - \mu^d (0)
\]
The right-hand side of the homotopy relation once again has only one term because the homotopy $H$ acts trivially on such elements, so we have that 
\[
((t\del)H + H(t\del))\mu^{d} = K \del \mu^{d}
\]
Note that $K\del \mu^{d}$ has vanishing constant term because of how $K$ was chosen and so does $\mu^d- \mu^d(0)$. Therefore, to show that they are equal, it suffices to show that their difference is constant, i.e. $\del$-closed. We have that 
\[
\del (\mu^d - \mu^d(0) - K\del \mu^d) = \del\mu^d - \del K \del \mu^d = \del\mu^d - \del \mu^d + K \del^2 \mu^d = 0.
\]
\end{itemize}

Now we analyze which diagrams are nonzero. Since the shifted Poisson structure vanishes along a smaller locus than in the case of minimal BCOV theory, the Hamiltonian vector field generated by the interaction can have more nonvanishing Taylor components. Indeed, we now see that there are Lie brackets that can output negative degree descendants, coming from those terms in the interaction which contain elements of the summand $\oplus_{1\leq i\leq d-k-1} t^i \PV^{d-k-i-1,\bullet}(X)$ However, there are no brackets that can output elements of this summand, as such a bracket would necessarily come from a term in the interaction which contains negative degree descendants.

We first claim that any diagram that contains the homotopy $H$ must evaluate to zero. Indeed, the homotopy is only nonzero on summands of the form $\bigoplus_{\substack{i > 0 \\ i+j\leq d-1\\ i+j\neq k}}t^{i}\PV^{j}$ and $\bigoplus_{0\leq i \leq d-k-2}t^{-i}\PV^{k+i+1}$ No brackets can output positive degree descendants so we need only consider the case where the diagram begins with a bracket that outputs an element of the first summand. The homotopy then maps such an element to a descendant with a negative power of $t$, which as remarked, does not occur as the input for a nonzero bracket. Therefore, only single vertex diagrams can be nonzero.

Next, we claim that any higher than cubic vertex that appears in a nonzero diagram must come from a term in the interaction that is linear in descendants, which must be from the summand $\oplus_{1\leq i\leq d-k-1} t^i \PV^{d-k-i-1,\bullet}(X)$. Indeed, consider a term in the action that is quadratic or higher in descendants -- note that these descendants must necessarily occur with positive powers of $t$. Such a term would give a Lie bracket that takes a descendant as input. However, descendants that appear with positive powers of $t$ are trivial in cohomology. 

Moreover, a term in the action that depends on $t^i$ can only be nonzero if it is of polynomial degree $i+3$. Therefore, using subscripts to simply enumerate the inputs, the nonzero higher brackets come from terms in the interaction of the following form. Letting $n = 0, \cdots \lfloor \frac{k+i+1}{k} \rfloor$, there is a term in the interaction given by 
\[
\operatorname{Tr}\langle t^i\mu^{d-k-i-1} \mu_1^{a_1}\cdots \mu_{i+2-n}^{a_{i+2-n}}\del(\mu^{k+1}_{i+2-n+1})\cdots \del(\mu^{k+1}_{i+2})\rangle_0
\]
where $a_1+\cdots + a_{i+2-n} = (1-n)k+i+1$. Such a term would produce a bracket whose output lands in $t^{-i}\PV^{k+i+1}$. However, the only summand with a negative power of $t$ on which the map $p$ is nonzero is that for which $i = d-k-1$ above. 

Therefore, in addition to quadratic brackets, we have the following $(d-k+1)$-ary brackets in cohomology. For each $n = 0, \cdots, \lfloor \frac dk\rfloor$, and $a_1, \cdots, a_{d-k+1-n}$ such that $a_1+\cdots + a_{d-k+1-n} = d-nk$, there is a $(d-k+1)$-ary bracket given by 
\begin{align*}
\PV^{a_1}_\del \times \cdots \times \PV^{a_{d-k+1-n}}_\del\times (\PV^{k+1}/\del\PV^{k+2})^{\times n} &\to \C \\
(\mu_1^{a_1}, \cdots ,\mu_{d-k+1-n}^{a_{d-k+1-n}}, \mu^{k+1}_{d-k+1-n+1}, \cdots, \mu^{k+1}_{d-k+1}) &\mapsto \left (\mu_1^{a_1}\cdots \mu_{d-k+1-n}^{a_{d-k+1-n}}\del(\mu^{k+1}_{d-k+1-n+1})\cdots \del(\mu^{k+1}_{d-k+1}) \right)(0)\vee \Omega \\
\end{align*}

\circled{3}
Finally, we show that the cohomology $H^\bullet (\wt {\mc{E}} _{\mr{mBCOV} }(\C^d); Q)$ is given by an $L_\infty$ central extension of $\Pi \C \oplus \SHO(d|d)$. 

We again consider the short exact sequence of Lie algebras. 
\[
0\to \ker\Phi_*\to  H^\bullet (\Pi \wt {\mc E}^k_{\mr {mBCOV}}(\C^d); Q) \stackrel{\Phi_*}{\to} H^\bullet (\Pi {\mc E}^k_{\mr {mBCOV}}(\C^d); Q)\to 0.
\]
The map $\Phi_*$ is induced from $\Phi$, which is in turn given by the $\del$ operator on the summand given by $\oplus_{0\leq i \leq d-k-1} t^{-i}\PV^{k+i+1,\bullet}$ and the identity elsewhere. By way of the holomorphic Poincare lemma, we can identify $\PV^{k+1}/\del \PV^{k+2} \cong   \PV^{k}_\del  $, so $\ker \Phi_*$ is again just a copy of constants -- if $d$ is even this sits in even degree and if $d$ is odd it sits in odd degree. 

Therefore, using the splitting gotten by $K^k|_{\ker\del}$, we have an isomorphism
\[
H^\bullet (\Pi \wt {\mc E}^k_{\mr {mBCOV}}(\C^d); Q) \cong H^\bullet (\Pi {\mc E}_{\mr {mBCOV}}(\C^d); Q) \oplus (\Pi) \C,
\]

by way of which the above brackets become certain $(d-k+1)$-cocycles. By Proposition \ref{prop:sho} we can identify $H^\bullet (\Pi {\mc E}_{\mr {mBCOV}}(\C^d); Q)\cong \SHO(d|d) \oplus \Pi \C$, so we find that $H^\bullet (\wt \Pi {\mc E}^k_{\mr {mBCOV}}(\C^d); Q)$ is a one-dimensional central extension of $\SHO(d|d) \oplus \Pi \C$. Moreover, note that the relevant cocycles are in fact not pulled back from cocycles on $\SHO(d|d)$ -- they depend nontrivially on the one-dimensional center.
\end{proof}

\section{Two $\mf{sl}_2$ actions}\label{sl2}
We conclude this article with an application of the above results. In prior work \cite{SuryaYoo}, we constructed an $\SL_2 \C$ action of minimal BCOV theory in three dimensions with $2$-potentials; a  similar action was identified in \cite{CostelloGaiotto}. By way of a conjectural description of the $\SU(3)$-invariant twist of type IIB supergravity in terms of this potential theory for BCOV, it was argued that this $\SL_2 \C$ action is a shadow of S-duality. Meanwhile, in \cite{Kac} it was shown that $\SHO(3|3)$ and an odd two dimensional central extension thereof have an $\mf{sl}_2$ action by outer derivations. In this section, we show that by way of theorem \ref{thm:bcov_d}, our $\SL_2 \C$ action infinitessimally agrees with Kac's action.

In the remainder of the article, we specialize to the case where $d=3$. We begin with a specialization of definition \ref{def:bcovpot}. 

\begin{ex}
  Let $X$ be a Calabi--Yau 3-fold. \textit{Minimal BCOV theory on $X$ with 2-potentials} $\wt{\mc E}_{\mr{mBCOV}^2}(X)$ is the Poisson degenerate BV theory with:
\begin{itemize}
\item underlying space of fields given by the sheaf of cochain complexes
\[\xymatrix @R=0.5em {
 \ul{-2} &   \ul{-1} &   \ul{0} \\
 \PV^{0,\bullet}(X) \\
& \PV^{1,\bullet}(X) \ar[r]^-{t\del } & t\PV^{0,\bullet}(X)\\
 &&\PV^{3,\bullet}(X)   }\]
\item shifted symplectic structure between $\PV^{0,\bullet}$ and $\PV^{3,\bullet}$ given by $\tr$ and shifted Poisson structure given by the kernel $(\del\otimes 1)\delta_{\mr{diag}}$ on $(\PV^{1,\bullet}(X)[1]\to t\PV^{0,\bullet}(X))$, and
\item an action functional given by
 \[\wt{I}_{{\text{mBCOV}}}(\alpha,\mu,t\nu, \gamma  )= \sum_{k\geq 0}  \tr \left( \alpha \mu  \del\gamma   \nu^k + \frac{1}{6} \mu\mu\mu  \nu^k   \right)  \]
where $\alpha \in \PV^{0,\bullet}(X)$, $(\mu,t\nu ) \in \PV^{1,\bullet}(X)\oplus t \PV^{0,\bullet}(X)$, and $\gamma\in \PV^{3,\bullet}(X)$. 
\end{itemize}
\end{ex}

As an immediate consequence of theorem \ref{thm:bcov_d}, we have the following corollary:

\begin{cor}
Minimal BCOV theory with 2-potentials on $\C^3$ is $L_\infty$-equivalent to an odd two-dimensional central extension of $\SHO(3|3)$. The cocycle is nonzero on certain constant coefficient and linear vector fields and explicitly gives the following Lie brackets
\begin{itemize}
\item $[\del_{x_i}, \xi_k \del_{x_j} - \xi_j\del_{x_k} ] = \eps_{ijk}e_1$
\item $[\del_{\xi_i}, \del_{x_j}] = \delta_{ij}e_2$
\end{itemize}
where $e_1, e_2$ is a basis for $\C^2$.
\end{cor}

In \cite{Kac} Kac computed the outer derivations of $\SHO(3|3)$ and found that it is $\mf{gl}_2$. Explicit formulae for this $\mf{gl}_2$ action are given in \cite{CantariniKac}; we begin by recalling this action.

To do so, we will make use of a grading on $\SHO(3|3)$ called the \textit{principal grading} \cite{Kac}. The grading is easily described using the identification of $\SHO(3|3)$ as nonconstant polynomials on $\C^{3|3}$ that are in the kernel of the divergence operator -- the degree $n$ summand consists of those polynomials of degree $n+2$. Accordingly, we see that $\SHO(3|3) = \bigoplus _{i\geq -1} \SHO(3|3)_{i}$ and that for $j \geq 1$, $\SHO(3|3)_{j}$ is given by $j$-fold brackets of elements of $\SHO(3|3)_1$

\begin{prop}[\cite{CantariniKac} \cite{CantariniKacErratum}]
The following assignments constitute a map $\mf{sl}_2\to \operatorname{Der}(\SHO(3|3))$:
\begin{itemize}
\item $e$ acts as $\operatorname{ad}(\xi_1\xi_3\del_{x_2}-\xi_2\xi_3\del_{x_1} - \xi_1\xi_2\del_{x_3})$
\item $h$ acts as $\operatorname{ad}(\sum_{i=1}^3 \xi_i \del_{\xi_i})$
\item $f$ acts as follows. It suffices to specify its action on the summands $\bigoplus_{i\geq -1}^1 \SHO(3|3)_i\subset \SHO (3|3)$. Using the description of $\SHO(3|3)$ as nonconstant polynomials on $\C^{3|3}$ with the Schouten--Nijenhuis bracket, we have 
\begin{align*}
f (\xi_i\xi_j) & = - \eps_{ijk}x_k \\
f(x_i\xi_j\xi_k) & = \frac 12 \eps_{ijk}x_i^2 \\
f(x_i \xi_i\xi_j - x_k\xi_k\xi_j) & = -\eps_{ijk}x_i x_k
\end{align*}
and as zero everywhere else on $\bigoplus_{i\geq -1}^1 \SHO(3|3)_i$
\end{itemize}
\end{prop}

The following statement is implicit in \cite{ChengKac}.

\begin{prop}
The above $\mf{sl}_2$ action extends by the standard representation to the odd two-dimensional central extension.
\end{prop}
\begin{proof}
It suffices to show that the cocycle is equivariant

\begin{itemize}
\item We first check the action of $h$. We have that 
\begin{align*}
[h \del_{x_i},\xi_k \del_{x_j} - \xi_j\del_{x_k}] + [\del_{x_i}, h (\xi_k \del_{x_j} - \xi_j\del_{x_k})] = [\del_{x_i},  \xi_k \del_{x_j} - \xi_j\del_{x_k}] = \eps_{ijk}e_1 = h(\eps_{ijk} e_1)
\end{align*}
Similarly,
\begin{align*}
[h \del_{\xi_i }, \del_{x_j}] + [ \del_{\xi_i }, h \del_{x_j}]  = [- \del_{\xi_i }, \del_{x_j}]  = -\delta_{ij} e_2  = h (\delta_{ij}e_2).
\end{align*}

\item Next for the action of $e$, we have that 
\begin{align*}
[e \del_{x_i},\xi_k \del_{x_j} - \xi_j\del_{x_k}] + [\del_{x_i}, e (\xi_k \del_{x_j} - \xi_j\del_{x_k})] = 0= e (\eps_{ijk}e_1)\\
\end{align*}
Similarly,
\begin{align*}
[e \del_{\xi_i }, \del_{x_j}] + [ \del_{\xi_i }, e \del_{x_j}]  &= [\eps_{ibc}\xi_b \del_{x_c}, \del_{x_j}]  =  \left[\del_{x_j}, \frac12 \eps_{ibc}(\xi_c \del_{x_b} - \xi_b \del_{x_c})\right ]\\
& = \frac12 \eps_{ibc}\eps_{jbc} e_1 = \delta_{ij} e_1  = e(\delta_{ij} e_2)
\end{align*}

\item Finally, for the action of $f$, we must rewrite the elements of $\SHO(3|3)$ for which the cocycle is nonzero in terms of functions on $\C^{3|3}$. In particular, since $f$ is only nonzero on certain quadratic functions, we need only worry about the term in the cocycle involving the linear vector field $\xi_k \del_{x_j} - \xi_j\del_{x_k}$ whose Hamiltonian generator is the quadratic function $\xi_j\xi_k$.

\begin{align*}
[f \del_{x_i},\xi_k \del_{x_j} - \xi_j\del_{x_k}] + [\del_{x_i}, h (\xi_k \del_{x_j} - \xi_j\del_{x_k})] = -\eps_{jkc} [\del_{x_i}, \del_{\xi_c}] = \eps_{jkc} \delta_{ic} e_2 = f(\eps_{ijk} e_1)
\end{align*}
Similarly,
\begin{align*}
[f \del_{\xi_i }, \del_{x_j}] + [ \del_{\xi_i }, f \del_{x_j}]  = 0 = f(\delta_{ij} e_2)
\end{align*}
\end{itemize}
\end{proof}

We claim that these outer derivations of $\SHO(3|3)$ and its odd two-dimensional central extension actually come from a symmetry of minimal BCOV theory with potentials in three complex dimensions. In particular, preserving the shifted Poisson structure on the space of fields picks out $\mf{sl}_2\subset \mf{gl}_2$. 

To see this, it will be useful to rewrite minimal BCOV theory with potentials in the following way.

\begin{lemma}
As a $\Z/2$-graded Poisson BV theory, minimal BCOV theory on $X$ with potentials is to a theory  $\wt {\mc E}^{\ \Z/2}_{\mr{mBCOV}}(X)$ with	
\begin{itemize}
\item underlying space of fields given by the sheaf of cochain complexes
\[\xymatrix @R=0.5em {
  \ul{\text{odd}} &   \ul{\text{even}} \\
& \PV^{0,\bullet}(X) \otimes \C^2 \\
 \PV^{1,\bullet}(X) \ar[r]^-{t\del } & t\PV^{0,\bullet}(X)    }\]\item shifted symplectic structure on $\PV^{0,\bullet}(X) \otimes \C^2 $ given by $\tr\circ ( (-)\vee  \Omega_X^{-1}) \circ \omega $, where $\omega$ is the canonical symplectic form on $\C^2$, and shifted Poisson structure given by the kernel $(\del\otimes 1)\delta_{\mr{diag}}$ on $(\PV^{1,\bullet}(X)[1]\to t\PV^{0,\bullet}(X))$, and
\item an action functional given by
 \[\wt{I}_{{\text{mBCOV}}}^{\ \Z/2}(\phi, \mu,t\nu   )= \sum_{k\geq 0}  \tr \left( \frac12 \omega(\phi, [\mu,\phi]_{\mr{SN}})  \nu^k \vee \Omega_X^{-1}  + \frac{1}{6} \mu\mu\mu  \nu^k   \right)  \]
where $\phi\in \PV^{0,\bullet}(X) \otimes \C^2 $ and $(\mu,t\nu ) \in \PV^{1,\bullet}(X)\oplus t \PV^{0,\bullet}(X)$.
\end{itemize}
\end{lemma}

\begin{proof}
Letting $e_1, e_2$ be a Darboux basis of $\C^2$, we can expand $\phi$ as $\phi = \alpha e_1 + \beta e_2$ where $\alpha, \beta \in \PV^{0,\bullet}(X)$. There's an obvious quasi-isomorphism of complexes $\Psi: \wt {\mc E}_{\mr{mBCOV}}(X) \to E_{\mr{mBCOV}}(X)$ given by $\beta \mapsto \gamma = \beta\vee \Omega^{-1}_X$.  

Now we have that \[\Psi^* \wt I_{\mr{mBCOV}} = \sum_{k\geq 0}  \tr \left( \alpha \mu  \del ( \beta  \vee \Omega_X^{-1} )   \nu^k + \frac{1}{6} \mu\mu\mu  \nu^k   \right) . \]
Now from the identity  \[\mu \del(\beta \vee \Omega_X^{-1}) = \mu (\del_{\mr{dR}}\beta \vee \Omega_X^{-1} )  = \langle \mu,  \del_{\mr{dR}} \beta\rangle  \vee \Omega_X^{-1} = [\mu,\beta ]_{\mr{SN}}  \vee \Omega_X^{-1}\]
we see that the above is in fact equal to  $\wt I_{\mr{mBCOV}}^{\ \Z/2}$.
where $\del_{\mr{dR}}$ is the holomorphic de Rham differential on $\Omega^{0,\bullet}(X) = \PV^{0,\bullet}(X)$.
\end{proof}

There is a manifest $\mf{sl}_2$ action on the resulting theory.

\begin{prop}\label{sl2bcov}
There is a Hamiltonian action of $\mf{sl}_2$ on $\wt {\mc E}^{\ \Z/2}_{\mr{mBCOV}}(X)$, that is, there is a map of DG Lie algebras
\begin{align*}
J\colon  (\mf{sl}_2 , d=0, [-,-]) & \longrightarrow  (\mc O (\wt {\mc E}^{\ \Z/2}_{\mr{mBCOV}}(X))[1] , \{\wt{I}^{\ \Z/2}_{{\text{mBCOV}}},-\} , \{-,-\}  )\\
 x & \longrightarrow \frac12 \int \omega (\phi, x \phi)\wedge \Omega_X
\end{align*}
\end{prop}

Note that we have exhibited an $\mf{sl}_2$ action by quadratic Hamiltonians. This determines an action by linear vector fields, which are therefore derivations of the shifted tangent complex. It therefore makes sense to compare the $\mf{sl}_2$ action on $\wt{\mc E}^{\ \Z/2}_{\mr{mBCOV}}(\C^3)$ with the $\mf {sl}_2$ action by derivations on the two-dimensional odd central extension of $\SHO(3|3)$. 

The following straightforward, if lengthy computation shows that the two $\mf{sl}_2$ actions are one and the same.

\begin{thm}\label{sl2equiv}

The composition 
\[
\SHO(3|3) \oplus \Pi \C^2 \cong H^\bullet (\wt{\mc E}_{\mr{mBCOV}}(\C^3)) \to  \wt{\mc E}_{\mr{mBCOV}}(\C^3) \cong  \wt{\mc E}^{\ \Z/2}_{\mr{mBCOV}}(\C^3)
\]
is $\mf{sl}_2$ equivariant. 
\end{thm}

Proofs of proposition \ref{sl2bcov} and theorem \ref{sl2equiv} will appear in a subsequent version of the present article.

\printbibliography

\end{document}